\documentclass{article}
\usepackage[utf8]{inputenc}
\usepackage[margin=1in]{geometry}
\usepackage{amssymb,amsfonts,amsmath,amsthm,comment,mleftright,enumerate,graphicx,hyperref,xcolor,siunitx,diagbox,stmaryrd,setspace,MnSymbol}
\usepackage[makeroom]{cancel}
\usepackage[title]{appendix}
\usepackage[round]{natbib}
\usepackage[noend]{algpseudocode}
\usepackage{algorithm,algorithmicx}

\newcommand*\Let[2]{\State #1 $\gets$ #2}
\algrenewcommand\algorithmicrequire{\textbf{Precondition:}}
\algrenewcommand\algorithmicensure{\textbf{Postcondition:}}

\newtheorem{lemma}{Lemma}

\newcommand{\paren}[1]{\left(#1\right)}
\newcommand{\bracket}[1]{\left[#1\right]}
\renewcommand{\brace}[1]{\left\{#1\right\}}
\renewcommand{\ang}[1]{\left\langle#1\right\rangle}

\newcommand{\R}{\mathbb{R}}
\newcommand{\C}{\mathbb{C}}
\newcommand{\F}{\mathbb{F}}
\newcommand{\Z}{\mathbb{Z}}
\newcommand{\ring}{\mathcal{R}}

\newcommand{\M}[1]{\begin{bmatrix}#1\end{bmatrix}}
\newcommand{\MA}[2]{\bracket{\begin{array}{#1}#2\end{array}}}

\newcommand{\rk}[1]{\mathrm{rk}\paren{#1}}
\newcommand{\brk}[2]{\underline{\mathrm{rk}}_{#2}\paren{#1}}
\newcommand{\borderk}[1]{\underline{\mathrm{rk}}\paren{#1}}

\renewcommand{\dim}[1]{\mathrm{dim}\paren{#1}}
\newcommand{\rref}[1]{\mathrm{rref}\paren{#1}}
\newcommand{\GL}[2]{\mathrm{GL}\paren{#1,\ #2}}
\renewcommand{\~}[1]{\widetilde{#1}}

\renewcommand{\Vec}[1]{\mathrm{vec}\paren{#1}}
\renewcommand{\span}{\mathrm{span}}

\newcommand{\cpdeval}[1]{\bracket{\!\bracket{#1}\!}}
\newcommand{\unfold}[2]{#1_{\paren{#2}}}

\makeatletter
\def\@biblabel#1{}
\makeatother

\usepackage{etoolbox}
\makeatletter
\pretocmd\@maketitle
  {\def\@makefnmark{\hbox{\@textsuperscript{\normalfont\@thefnmark}}}}
  {}{\FAIL}
\makeatother

\title{Depth-first search for tensor rank and border rank over finite fields}
\author{Jason Yang}
\date{}

\begin{document}

\maketitle

\begin{abstract}
We present an $O^*\paren{|\F|^{(R-n_*)\paren{\sum_d n_d}+n_*}}$-time algorithm for determining whether a tensor of shape $n_0\times\dots\times n_{D-1}$ over a finite field $\F$ has rank $\le R$, where $n_*:=\max_d n_d$; we assume without loss of generality that $\forall d:n_d\le R$.
We also extend this problem to its border rank analog, i.e., determining tensor rank over rings of the form $\F[x]/(x^H)$, and give an $O^*\paren{|\F|^{H\sum_{1\le r\le R} \sum_d \min(r,n_d)}}$-time algorithm.
Both of our algorithms use polynomial space.
\end{abstract}

\section{Introduction}
Over an arbitrary ring $\ring$ (i.e., a set with addition, multiplication, additive identity, and multiplicative identity, but not necessarily division), a \textit{rank-$R$ canonical polyadic decomposition (CPD)} of a tensor (i.e., multidimensional array) $T\in \ring^{n_0\times\dots\ n_{D-1}}$, is a list of \textit{factor matrices} $A_d\in\ring^{n_d\times R},\ 0\le d<D$, such that \[\forall i_0,\dots,i_{D-1}:\ T_{i_0,\dots,i_{D-1}} = \sum_{0\le r<R} \prod_{0\le d<D} (A_d)_{i_d,r}.\]

This relation is commonly notated as $T=\cpdeval{A_0,\dots,A_{D-1}}$ or $T=\sum_{0\le r<R} \bigotimes_d (A_d)_{:,r}$. The rank of the tensor $T$ is the smallest $R$ such that there exists a rank-$R$ CPD of $T$.

Tensor rank is the central problem underlying fast matrix multiplication (MM), the main performance bottleneck of most linear algebra operations, backpropagation, various graph algorithms, and many other problems. 
Formally, the action of multiplying a $m\times k$ matrix with a $k\times n$ matrix can be described with a $mk\times kn\times nm$-shaped tensor $\mathcal{M}_{\ang{m,k,n}}$; then a rank-$R$ CPD of this tensor yields an $O\paren{N^{3\log_{mkn} R}}$-time algorithm for multiplying two $N\times N$ matrices, for arbitrarily large $N$ \citep{blaser}.
Thus, finding CPDs of $\mathcal{M}_{\ang{m,k,n}}$ tensors with smaller rank than previously known would lead to faster algorithms for MM, which makes tensor rank an important problem.

Usually $\ring$ is the real numbers ($\R$) or complex numbers ($\C$), but setting it to a finite field can also be useful. For example, one can find a CPD of a tensor over a small finite field (often the integers mod 2), then lift to a larger ring such as $\Z$; this is a common technique in fast MM \citep{alphatensor, flip}, since CPDs over the integers are the most desired, for reasons such as reducing floating-point errors.
Additionally, restricting $\ring$ to a finite field gives some hope that obtaining an CPD is even possible, which other rings do not share: for example, a CPD using real numbers usually cannot be represented exactly; and finding tensor rank over the integers is known to be undecidable \citep{shitov}. For these reasons, we are interested in working over finite fields.

We also want to broaden our scope to analyzing arbitrary tensors, not just MM tensors, as we want to know if general algorithmic speedups exist. Additionally, tensor rank has applications outside of fast MM, such as the light bulb problem \citep{light} and dynamic programming \citep{tensor-dp}.

Thus, our goal is to solve the following problem: given a finite field $\F$, an arbitrary tensor $T\in\F^{n_0\times\dots\times n_{D-1}}$, and rank threshold $R$, return a CPD of $T$ of rank $\le R$ if one exists, or prove that no such CPD exists.
Although this problem is NP-complete \citep{hastad}, it may still be possible and useful to attain exponential speedup over brute force.

We require our algorithm to have guaranteed correctness, which most previous work does not satisfy: for example, \citep{courtois, heule, flip} use random search, and \citep{alphatensor} uses reinforcement learning, none of which can prove that a CPD of specific rank does not exist.

To this end, we present an $O^*\paren{|\F|^{(R-n_*)\paren{\sum_d n_d}+n_*}}$-time polynomial-space algorithm for this problem, where $O^*$ suppresses polynomial factors and $n_*:=\max_d n_d$; here we assume WLOG that $\forall i: n_i\le R$. For fixed $R$ the time complexity is bounded above by $O^*\paren{|\F|^{
D\cdot \frac{(R+\frac{1}{D})^2}{4}
}}$.

We also investigate \textit{border} CPDs of exponent $H$, which are CPDs of a tensor $x^{H-1} T$ over the ring $\F[x]/(x^H)$, where all elements in $T$ are in $\F$; such CPDs are also widely used in fast MM (more detail in Section \ref{border}).
For this problem, we present an $O^*\paren{|\F|^{H\sum_{1\le r\le R} \sum_d \min(r,n_d)}}$-time polynomial-space algorithm; for fixed $R$ the time complexity is bounded above by $O^*\paren{|\F|^{HD \cdot \frac{R(R+1)}{2}}}$.

The code for our algorithms is available at \url{https://github.com/coolcomputery/tensor-cpd-search}.

\subsection{Notation}
\begin{itemize}
    \item Tensor slices are denoted with NumPy notation.
    \begin{itemize}
        \item for a two-dimensional tensor $A$: $A_{r,:}$ is the $r$-th row of $A$; $A_{:,:c}$ denotes the submatrix containing the first $c$ columns of $A$.
    \end{itemize}

    \item The flattening of a tensor $T$ is denoted as $\Vec{T}$, where elements are ordered in some consistent manner (e.g., row-major).

    \item The axis-$d$ (``mode-$d$") unfolding of a tensor $T$ is $\unfold{T}{d}:=\MA{c}{\vdots\\\hline \Vec{T_{:,\dots,:, i_a, :,\dots,:}} \\\hline \vdots}_{i_a}$.
    
    \item The tensor product of tensors $A$ and $B$
    is $A\otimes B:=\M{A_{i_0,\dots,i_{D-1}} B_{j_0,\dots,j_{D'-1}}}_{i_0,\dots,i_{D-1}, j_0,\dots,j_{D'-1}}$.
    
    \item The axis-$d$ operation of $T\in\ring^{n_0\times\dots\times n_{D-1}}$ by a matrix $M\in\ring^{n'_d \times n_d}$
    is
    
    $M\times_d T:=
    \M{\sum_{i_d} M_{i'_d,i_d} T_{_{i_0,\dots,i_{D-1}}}}_{i_0,\dots,i_{d-1}, i'_d, i_{d+1},\dots,i_{D-1}}$.


    \item A CPD with rank $R$ (``rank-$R$ CPD") is a list of factor matrices $A_d\in\ring^{n_d\times R},\ 0\le d<D$ that evaluates to the tensor $\cpdeval{A_0,\dots,A_{D-1}}:=\M{\sum_{0\le r<R} \prod_d (A_d)_{i_d,r}}_{i_0,\dots,i_{D-1}}
    =\sum_r \bigotimes_d (A_d)_{:,r}$.

    \item The rank of a tensor $T$ is the smallest $R$ such that $T$ has a rank-$R$ CPD $\cpdeval{A_0,\dots,A_{D-1}}$, and is denoted as $\rk{T}$. When $T$ is a matrix, this definition coincides with matrix rank.

    \item The phrase ``rank-$\le R$ CPD" is short for ``a CPD of rank at most $R$."

    \item The reduced row echelon form of a matrix $M$ is denoted as $\rref{M}$.

    \item $\GL{n}{\F}$ denotes the set of invertible $n\times n$ matrices with elements in a field $\F$.
\end{itemize}

\section{Axis-reduction}
\label{axis-reduce}
Na\"ively, finding a rank-$\le R$ CPD of a tensor $T\in\F^{n_0\times\dots\times n_{D-1}}$ using brute force would take $O^*\paren{|\F|^{R\sum_d n_d}}$ time, as such a CPD $\cpdeval{A_0,\dots,A_{D-1}}$ contains $Rn_0+\dots+Rn_d$ many variables, each with $|\F|$ many possible values.

By itself, this running time is not very good because each $n_i$ could be arbitrarily large, even if $R$ is small.
Here, we show that one can always transform $T$ so that $\forall d:n_d\le R$ without affecting the existence of a solution.

For each axis $d$, define the following:
\begin{itemize}
    \item the axis-$d$ unfolding of $T$, denoted as $\unfold{T}{d}$, a matrix whose $i$-th row is the flattening of the slice $T_{\underbrace{:,\dots,:}_d,i,:,\dots,:}$, where all slices are flattened in the same manner
    
    \item an arbitrary change-of-basis matrix $Q_d\in\GL{n_d}{\F}$ such that $Q_d \unfold{T}{d} = \rref{\unfold{T}{d}}$
    
    \item the axis-$d$ rank of $T$, $r_d:=\rk{\unfold{T}{d}}$
\end{itemize}

Then we create the new tensor $T'=(Q_{D-1})_{:r_{D-1},:}\times_{D-1} \paren{\dots \paren{(Q_0)_{:r_0,:}\times_0 T}}$, which has shape $r_0\times\dots\times r_{D-1}$. We call this process \textit{axis-reduction}; note that each $Q_d$ can be found using Gaussian elimination, so axis-reduction runs in polynomial time.


We have $\rk{T}=\rk{T'}$, since $T=(Q_0^{-1})_{:,:r_0}\times_0 \paren{\dots\paren{(Q_{D-1}^{-1})_{:,:r_{D-1}}\times_{D-1} T}}$ and doing axis operations $\times_d$ does not increase CPD rank.
In fact, any CPD of $T=\cpdeval{A_0,\dots,A_{D-1}}$, can be transformed into an equal-rank CPD of $T'=\cpdeval{(Q_0)_{:r_0,:}A_0,\dots,(Q_{D-1})_{:r_{D-1},:}A_{D-1}}$, and vice versa; so we can search for a CPD over $T'$ and then transform it back to a CPD for $T$.

The upshot is that if $T$ has a rank-$\le R$ CPD $\cpdeval{A_0,\dots,A_{D-1}}$, it must satisfy $\forall d:r_d\le R$: this is because each axis-$d$ slice $T_{:,\dots,:,i,:,\dots,:}$ would equal $\sum_r (A_d)_{i,r}\paren{\bigotimes_{d'\ne d} (A_{d'})_{:,r}}$, so $\span\brace{T_{:,\dots,:,i,:,\dots,:}}_{0\le i<n_d} \subseteq \span\brace{\bigotimes_{d'\ne d} (A_{d'})_{:,r}}_{0\le r<R}$. Therefore, if any $r_d$ exceeds $R$, we immediately know that a rank-$R$ CPD of $T$ does not exist.

Already, axis-reduction ensures that brute force runs within $O^*\paren{|\F|^{DR^2}}$ time, because after replacing $T$ with its axis-reduced version, every $n_d$ is at most $R$.
A more clever observation is that if we fix all first columns $(A_d)_{:,0}$ of the factor matrices, then $T-\bigotimes_d (A_d)_{:,0}$ must have rank $\le R-1$, so we can apply axis-reduction again and recurse in a depth-first search fashion.
Because $\forall d: n_d\le R$ after axis-reduction, the root call, which has rank threshold $R$, has at most $|\F|^{DR}$ many children calls, each of which has rank threshold $R-1$.
Continuing this observation all the way down to calls with rank threshold 0, we end up making $O\paren{|\F|^{DR}\cdots |\F|^{D\cdot 1}}=O\paren{|\F|^{D(R+\dots+1)}}=O\paren{|\F|^{D\cdot \frac{R(R+1)}{2}}}$ total calls. Since each call takes polynomial work, this depth-first-search algorithm runs within $O^*\paren{|\F|^{D(R+\dots+1)}}=O^*\paren{|\F|^{D\cdot \frac{R(R+1)}{2}}}$ time, significantly faster than brute force.
We use essentially the same algorithm to search for border CPDs in Section \ref{border}.

Finally, we note that axis-reduction can efficiently find rank-$\le 1$ CPDs, using the fact that $1^D=1$:
\begin{lemma}
\label{rk1}
For any tensor $T\in\F^{n_0\times\dots\times n_{D-1}}$, $\rk{T}\le 1$ if and only if $\forall d: \rk{\unfold{T}{d}}\le 1$. Furthermore, if this condition is satisfied, a minimum-rank CPD of $T$ can be found in polynomial time.
\end{lemma}

\section{Tensor rank}
Surprisingly, for regular CPD, it is possible to achieve our desired running time without explicitly using recursion, while only using axis-reduction once.



\label{rref-trick}

\subsection{Algorithm}
Assume that $T\in\F^{n_0\times\dots\times n_{D-1}}$ is axis-reduced and suppose a rank-$R$ CPD $T=\cpdeval{A_0,\dots,A_{D-1}}$ exists; note that this implies $\forall d: n_d\le R$.

Construct some $Q\in\GL{n_0}{\F}$ such that $QA_0=\rref{A_0}$. Because $\unfold{T}{0}=A_0\MA{c}{\vdots\\\hline \Vec{\bigotimes_{d'\ne 0} (A_{d'})_{:,r}}\\\hline \vdots}_r$, rank bounds imply that $\rk{A_0}=n_0$, so we can assume $\rref{A_0}=\MA{c|c}{I_{n_0} & X}$ for some arbitrary matrix $X$, after simultaneously permuting the columns of all factor matrices appropriately.

Denote $\~{T}:=T - \cpdeval{(A_0)_{:,n_0:},\cdots,(A_{D-1})_{:,n_0:}}$; then $Q\times_0 \~{T}=\cpdeval{I_{n_0},(A_0)_{:,:n_0},\cdots,(A_{D-1})_{:,:n_0}}$. For brevity, denote the right-hand side as $\Xi$.



Suppose $\~{T}$ is fixed: to solve for $Q$, note that every axis-$0$ slice of $\Xi$ has rank $\le 1$, but otherwise, we have complete freedom over $\Xi$ with respect to $\~{T}$. Thus, it is necessary and sufficient for the rows of $Q$ to be linearly independent and to each satisfy $\rk{Q_{i,:}\times_0 \~{T}}\le 1$ for all $i$. We can find such a set of rows by extracting a basis subset of $S:=\brace{v\in\F^{1\times n_0} : \rk{v\times_0\~{T}}\le 1}$ using Gaussian elimination, or determining that such a basis does not exist.

If we do find a satisfying $Q$, then by Lemma \ref{rk1} we can find a rank-$\le 1$ CPD of each $Q_{i,:}\times_0 \~{T}$ in polynomial time; join these CPDs and transform the first factor matrix by $Q^{-1}$ to create a rank-$\le R$ CPD of $\Xi$; and join with $\paren{(A_0)_{:,n_0:},\cdots,(A_{D-1})_{:,n_0:}}$ to create a rank-$\le R$ CPD of $T$. All that remains is to enumerate all possible tensors $\~{T}$, which can be done by enumerating $\paren{(A_0)_{:,n_0:},\cdots,(A_{D-1})_{:,n_0:}}$.






\subsection{Complexity}

There are $|\F|^{(R-n_0)\paren{\sum_d n_d}}$ many possible assignments of $\paren{(A_0)_{:,n_0:},\cdots,(A_{D-1})_{:,n_0:}}$ to fix, and thus that many $\~{T}$ to iterate over.

For each $\~{T}$, when we solve for $Q$: enumerating each $v$ takes $|\F|^{n_0}$ time; checking $\rk{v\times_0\~{T}}\le 1$ takes polynomial time by Lemma \ref{rk1}; and $\dim{\span\brace{S}}=\rk{\MA{c}{\vdots\\\hline v \\\hline \vdots}_{v\in S}}$, which can be calculated in polynomial time using Gaussian elimination. Thus, this subroutine takes $O^*\paren{|\F|^{n_0}}$ time.

Overall, our algorithm runs in time $O^*\paren{|\F|^{(R-n_0)\paren{\sum_d n_d}+n_0}}$.
We can improve this time complexity by permuting the axes of $T$ so that $n_0$ is the longest shape length of $T$, allowing us to replace $n_0$ in the expression with $n_*:=\max_d n_d$.

The pseudocode for this algorithm is shown in Algorithm \ref{rref-alg}, which is implemented to consume polynomial space. One can modify the algorithm to more closely resemble depth-first search, e.g., by fixing tuples of columns $\paren{(A_0)_{:,r},\cdots,(A_{D-1})_{:,r}}$ for each $n_0\le r<R$ at a time.

\begin{algorithm}
    \setstretch{1}
    \caption{Search algorithm for rank-$\le R$ CPD over a finite field $\F$}
    \label{rref-alg}
    \begin{algorithmic}[1]
    \newcommand{\target}{\mathrm{target}}
    \Require{
        $T\in\F^{r_0\times\dots\times r_{D-1}}$, $T$ axis-reduced,
        $R\in\Z_{\ge 0}$
    }
    \Ensure{
        returns a rank-$\le R$ CPD of $T$, if one exists; else, returns null
    }
    \Function{rref\_search\_help}{$T, R$}
        \For{$0\le d'<D,\ V_{d'}\in\F^{r_{d'}\times (R-r_0)}$}
            \Let{$\~{T}$}{$T-\cpdeval{V_0,\dots,V_{D-1}}$}
            \For{$0\le d<D$}
                \Let{$U_d$}{$[]$}
            \EndFor
            \Let{$Q$}{$[]$}
            \Let{$x$}{0}
            \For{$v\in\F^{1\times r_0}$}
                \Let{$t, u_1,\dots,u_{D-1}, s_1,\dots,s_{D-1}$}{\Call{axis\_reduce}{$(v\times_0 T)_{0,:,\dots,:}$}}
                \If{$(\forall d>1: s_d\le 1)$ \textbf{and} $\rk{\MA{c}{Q\\\hline v}}>\rk{Q}$}
                    \Let{$Q$}{$\MA{c}{Q\\\hline v}$}
                    \If{$\forall d>1: s_d>0$}
                        \Let{$u_0$}{$t_{0,\dots,0}e_x$} \Comment{$e_x$ is the $n_0$-long vector with a 1 at index $x$ and 0s everywhere else}
                        \For{$0\le d<D$}
                            \Let{$U_d$}{$\MA{c|c}{U_d & u_d}$}
                        \EndFor
                    \EndIf
                    \Let{$x$}{$x+1$}
                \EndIf
            \EndFor
            \If{$x=n_0$}
                \State \Return{$\paren{
                    Q^{-1}\MA{c|c}{U_0&V_0},
                    \MA{c|c}{U_1&V_1},
                    \dots,
                    \MA{c|c}{U_{D-1}&V_{D-1}}
                }$}
            \EndIf
        \EndFor
      \State \Return{\textbf{null}}
    \EndFunction

    \Statex
    \Require{
        $T_\target\in\F^{n_0\times\dots\times n_{D-1}}$,
        $R\in\Z_{\ge 0}$
    }
    \Ensure{
        returns a rank-$\le R$ CPD of $T_\target$, if one exists; else, returns null
    }
    \Function{rref\_search}{$T_\target, R$}
        \Let{$T, Q_0,\dots,Q_{D-1}, r_0,\dots,r_{D-1}$}{\Call{axis\_reduce}{$T_\target$}}
        \If{$\exists d | r_d>R$}
            \State \Return{\textbf{null}}
        \EndIf
        \If{$\exists d | r_d=0$}
            \State \Return{$\paren{\M{\ }}_{0\le d<D}$}
        \EndIf
        \Let{$d^*$}{$\arg\max_d r_d$}
        \Let{$A$}{\Call{rref\_search\_help}{$\M{T_{i_0,\dots,i_{D-1}}}_{i_{d^*},i_1,\dots,i_{d^*-1},i_0,i_{d^*+1},\dots,i_{D-1}},\ R$}} \Comment{swap axes 0 and $d^*$}
        \If{$A\ne\textbf{null}$}
            \For{$0\le d<D$}
                \Let{$U_d$}{$\paren{Q_d^{-1}}_{:,:r_d} A_d$}
            \EndFor
            \State \Return{ $\paren{
                U_{d^*},
                U_1,\dots,U_{d^*-1},
                U_0,
                U_{d^*+1},\dots,U_{D-1}
            }$ }
        \EndIf
        \State \Return \textbf{null}
    \EndFunction
  \end{algorithmic}
\end{algorithm}

\section{Border rank}
\label{border}
\subsection{Motivation and Definitions}
Intuitively, a border CPD is a parameterized CPD that approximates a target tensor $T$ arbitrarily well. Such a CPD may have a strictly smaller rank than $\rk{T}$: a famous example is $T=\M{\M{0&1\\1&0}&\M{1&0\\0&0}}$, which has rank 3 but is approximated by the rank-2 border CPD $\frac{1}{x}\paren{\M{1\\x}^{\otimes 3} - \M{1\\0}^{\otimes 3}}$ as $x\rightarrow 0$.


To formally define border CPDs, we multiply everything by a high enough power of $x$ so that everything is a polynomial in $x$. Define the following for a tensor with elements in $\F$ \citep{blaser}:
\begin{itemize}
    \item $\brk{T}{H} := \rk{x^{H-1} T}$ over the polynomial ring $\F[x]/(x^H)$; we call $H$ the \textit{exponent threshold}.
    \item $\borderk{T} := \min_{H\ge 1} \brk{T}{H}$; this is the \textit{border rank} of $T$.
\end{itemize}

Border rank is useful in fast MM because one can essentially pretend it is regular rank for time complexity purposes.
Recall from the introduction that finding a rank-$R$ CPD of the $\mathcal{M}_{\ang{m,k,n}}$ tensor leads to a $O\paren{N^{3\log_{mkn} R}}$-time algorithm for multiplying two $N\times N$ matrices;
it turns out this is also true for a \textit{border} rank-$R$ CPD (with arbitrary exponent threshold $H$), up to an extra $O\paren{N^\varepsilon}$-factor in the running time for arbitrarily small $\varepsilon>0$ \citep{blaser}.
Because border rank allows extra freedom over regular tensor rank, due to multiples of $x^H$ vanishing, most tensors (including MM tensors) usually have border CPDs with strictly lower rank than their tensor rank, making it easier to develop asymptotically fast algorithms for MM.

\subsection{Border depth-first search}
To make the search problem finite, we fix the exponent threshold $H$. We also generalize the problem to finding the CPD rank of a tensor with elements in the ring $\mathcal{R}:=\F[x]/(x^H)$, i.e., elements are not restricted to $\F$-multiples of $x^{H-1}$. We call $\mathcal{R}$ the \textit{border ring on field $\F$ with exponent threshold $H$}.

We can use the same depth-first search algorithm described in Section \ref{axis-reduce} to search for border CPDs; all we need is an efficient border analog of axis-reduction. It suffices to generalize \textit{matrix rank factorization} to border rings:

\begin{center}
Given $M\in\mathcal{R}^{m\times n}$, find $U\in\mathcal{R}^{m\times r}, V\in\mathcal{R}^{r\times n}$ such that $M=UV$ and $r$ is minimized.
\end{center}

Before proceeding, we use the following lemma about multiplicative inverses:
\begin{lemma}
An element $\alpha\in\F[x](x^H)$ has a multiplicative inverse if and only if it is not a multiple of $x$. Furthermore, such an inverse is unique and can be found in $O(H^2)$ time.
\end{lemma}
\begin{proof}
Clearly, if $\alpha$ is a multiple of $x$ it cannot have a multiplicative inverse because 1 is not a multiple of $x$.
Otherwise, $\alpha=\sum_{0\le h<H} \alpha_h x^h$ for $\alpha_h\in\F$ and $\alpha_0\ne 0$.
Suppose $\beta:=\sum_{0\le h<H} \beta_h x^h,\ \beta_h\in\F$ satisfies $\alpha\beta=1$. Then $\beta_0=\frac{1}{\alpha_0}$ and $\forall h\ge 1: \beta_h=-\frac{1}{\alpha_0}\paren{\sum_{h'<h} \alpha_{h-h'}\beta_{h'}}$.
\end{proof}

We sketch the following procedure for row-reduction on $M$:
\begin{itemize}
    \item Initialize a row index $r\gets 0$.
    \item Find an invertible element of $M$; move it to row $r$ using a row swap, normalize the row, and denote this as a pivot; then eliminate all other elements in the same column as the pivot and increment $r$.
    \item Repeat this process, but now with the first row frozen, i.e., it cannot be modified or moved.
    \item If at some point we cannot find an invertible element, we instead try to find an element of the form $x\gamma$ for invertible $\gamma$, normalize it to $x$, and then only eliminate terms that are multiples of $x$; if we cannot find such an element, repeat for $x^2\gamma$, and so on.
    \item At the end, remove any all-zeros rows and set $V$ to this matrix.
    \item Throughout this entire process, update $U$ as we perform each row operation; at the end, keep only its leftmost $r$ columns.
\end{itemize}

$V$ will have a form similar to row echelon form (not necessarily \textit{reduced} row echelon form), up to permutation of columns, except that the leading terms are powers of $x$.

To prove that such $V$ has the minimum number of rows for matrix rank factorization, we first simplify it further with column operations. For each leading term $x^h$, all other elements on the same row must be multiples of $x^h$ because we could have only had $x^h$ as a leading term if the non-frozen rows had no more terms $x^{h'}\gamma$ for $h'<h$ and invertible $\gamma$. This property means that $V$ can be reduced to the form $\M{I_{b_0} \\ &\ddots \\ &&x^{H-1}I_{b_{H-1}}}$ using column operations, up to the presence of extra all-zeros columns.

We now prove that the rank of this reduced matrix is equal to the number of nonzeros on its diagonal \citep{austin}:
\begin{lemma}
\label{xpow-diag}
Over $\F[x]/(x^H)$, $\rk{x^h I_n}=n$ for all $h<H$.
\end{lemma}
\begin{proof}
Suppose there was a rank-$<n$ CPD of $x^h I_n=\cpdeval{A_0,A_1}=A_0 A_1^\intercal$. If we compute this CPD over the larger polynomial ring $\F[x]$, then $A_0 A_1^\intercal=x^h I_n + x^H X$ for some arbitrary $X\in\F[x]^{n\times n}$.

We have $\det\paren{A_0 A_1^\intercal}=0$, since row-reducing $A_0$ is guaranteed to yield at least one row of all zeros (due to having fewer columns than rows), so there exists a sequence of row operations on $A_0 A_1^\intercal$ that also produces an all-zeros row.

However, $\det\paren{x^h I_n + x^H X}=x^{hn} + o(x^{hn})$, where $o(x^p)$ hides terms containing $x$ raised to a power $>p$; this is because in the Leibniz formula of the determinant, the term corresponding to the main diagonal is $x^{hn}$, whereas all other terms have $x$ raised to a strictly higher power. Hence, a contradiction.
\end{proof}

\begin{lemma}
Over $\F[x]/(x^H)$, let $N:=\M{I_{b_0} \\ &\ddots \\ &&x^{H-1}I_{b_{H-1}}}$; then $\rk{N}=\sum_h b_h$.
\end{lemma}
\begin{proof}
Clearly $\rk{N}\le \sum_i b_i$. To show the lower bound, multiply each $x^h I_{b_h}$ block by $x^{H-1-h}$ so that all terms on the diagonal are $x^{H-1-h}$; doing so cannot increase the rank. Then by Lemma \ref{xpow-diag}, this matrix has rank $\sum_i b_i$.
\end{proof}

Pseudocode for border CPD search is provided in Algorithm \ref{border-dfs}.
To bound its running time, note that a call will only recurse if all axis-ranks $r_d$ are $\le R$; if so, it branches to $|\F|^{H\sum_d r_d}$ children that each replace the parameter $R$ with $R-1$. Thus, if we now define $r_d$ to be the axis-$d$ rank of the original target tensor $T$, the total running time is $O^*\paren{|\F|^{H\sum_{1\le r\le R} \sum_d \min(r, r_d)}}$.

\begin{algorithm}
    \caption{Depth-first search algorithm for rank-$\le R$ CPD over border ring $\F[x]/(x^H)$}
    \label{border-dfs}
    \begin{algorithmic}[1]
    \newcommand{\target}{\mathrm{target}}
    \Require{
        $T_\target\in\paren{\F[x]/(x^H)}^{n_0\times\dots\times n_{D-1}}$,
        $R\in\Z_{\ge 0}$
    }
    \Ensure{
        returns a rank-$\le R$ CPD of $T_\target$, if one exists; else, returns null
    }
    \Function{dfs}{$T_\target, R$}
        \Let{$T, Q_0,\dots,Q_{D-1}, r_0,\dots,r_{D-1}$}{\Call{border\_axis\_reduce}{$T_\target$}}
        \If{$\exists d | r_d>R$}
            \State \Return{\textbf{null}}
        \EndIf
        \If{$\exists d | r_d=0$}
            \State \Return{$\paren{\M{\ }}_{0\le d<D}$}
        \EndIf
        \For{$0\le d<D,\ u_d\in\paren{\F[x]/(x^H)}^{r_d}$}
            \Let{$A$}{\Call{dfs}{$T-\bigotimes_d u_d ,\ R-1$}}
            \If{$A\ne\textbf{null}$}
                \State \Return{ $\paren{\paren{Q_d^{-1}}_{:,:r_d} \MA{c|c}{A_d & u_d}}_{0\le d<D}$ }
            \EndIf
        \EndFor
      \State \Return{\textbf{null}}
    \EndFunction
  \end{algorithmic}
\end{algorithm}

\section{Conclusion and Next Steps}
We have presented an $O^*\paren{|\F|^{(R-\max_d n_d)\paren{\sum_d n_d}+\max_d n_d}}$-time algorithm for finding a rank-$\le R$ CPD (or proving such a CPD does not exist) of a tensor with shape $n_0\times\dots\times n_{D-1}$ and elements in a finite field $\F$,
as well as a $O^*\paren{|\F|^{H\sum_{1\le r\le R} \sum_d \min(r,n_d)}}$-time algorithm to solve the same problem for a tensor with elements in the ring $\F[x]/(x^H)$.
Both algorithms are substantially faster than brute force, with an important reason why being axis-reduction; both algorithms also use polynomial space.

There is a large gap between the time complexity for border CPD and what it could be if our algorithm for exact CPD generalized to the border analog. The main issue is that reduced row-echelon form does not always exist over $\F[x]/(x^H)$ (where we are only allowed row operations) because a leading term that is not equal to 1 may have elements above it that are not multiples of itself; an example is $\M{1&1\\&x}$ for $H\ge 2$.
A corollary is that a full-rank square matrix is not necessarily invertible over border rings, in stark contrast to fields.

One possible way to close this gap is to analyze what happens when $R=\max_d r_d$, where $r_d$ is the axis-$d$ rank of the original target tensor. In this situation, we suspect that only $o\paren{|\F|^{H\sum_d r_d}}$ many children of the root call in Algorithm \ref{border-dfs} will recurse into their own children, because most children calls fail to reduce any axis-ranks of the target tensor. Studying this problem further could result in tighter time bounds without changing the algorithm.

\section{Acknowledgments}
We thank Prof. Virginia Vassilevska Williams for her continual mentorship throughout the years on this research topic, and Austin Conner for resolving how to calculate the rank of a general matrix over a border ring. We also thank Prof. Erik Demaine and Ani Sridhar for motivating us to pivot our research from fast matrix multiplication tensors to general tensors.

\end{document}